\documentclass[letterpaper, 10 pt, conference]{ieeeconf}
\IEEEoverridecommandlockouts 

\usepackage{cite}
\usepackage{amsmath, amssymb, amsfonts}
\usepackage{graphicx}
\usepackage{textcomp}
\usepackage{xcolor}
\usepackage{algorithm}
\usepackage{algpseudocode}
\usepackage{tikz}
\usetikzlibrary{shapes,positioning,calc,arrows.meta}
\usepackage{longtable}
\usepackage{multirow}
\usepackage{etoolbox}
\usepackage{hyperref} 
\usepackage{subfig}
\usepackage{subcaption}
\usepackage{mathtools}

\usepackage{booktabs}

\usepackage{enumerate}

\usepackage{bm}

\makeatletter
\def\@opargbegintheorem#1#2#3{%
  \@IEEEtmpitemindent\itemindent
  \topsep 0pt\rmfamily\trivlist
  \item[\hskip\labelsep{\indent\itshape #1\ #2}]%
  \itemindent\@IEEEtmpitemindent
  {\itshape (#3): }\ignorespaces
}
\makeatother

\newtheorem{lemma}{Lemma}
\newtheorem{theorem}{Theorem}

\newtheorem{remark}{Remark}
\newtheorem{definition}{Definition}

\newtheorem{example}{Example}

\AtEndEnvironment{problem}{\hfill$\Box$}
\AtEndEnvironment{proposition}{\hfill$\Box$}
\AtEndEnvironment{lemma}{\hfill$\Box$}
\AtEndEnvironment{theorem}{\hfill$\Box$}
\AtEndEnvironment{assumption}{\hfill$\Box$}
\AtEndEnvironment{condition}{\hfill$\Box$}
\AtEndEnvironment{remark}{\hfill$\Box$}
\AtEndEnvironment{corollary}{\hfill$\Box$}
\AtEndEnvironment{definition}{\hfill$\Box$}
\AtEndEnvironment{example}{\hfill$\Box$}

\newcommand{\R}{\mathbb{R}}

\DeclareMathOperator{\tr}{tr}

\providecommand{\SO}{\ensuremath{\mathbf{SO}}}

\begin{document}
\title{\LARGE \bf Complementary Filtering on \SO(3) for Attitude Estimation with Scalar Measurements}

\author{A.~Melis, S.~Berkane,~\IEEEmembership{Senior Member,~IEEE}, R.~Mahony,~\IEEEmembership{Fellow,~IEEE},
and T.~Hamel,~\IEEEmembership{Fellow,~IEEE}
\thanks{Alessandro Melis and Tarek Hamel are with I3S, CNRS, Université Côte d'Azur, Sophia Antipolis, France. Tarek Hamel is also with the Institut Universitaire de France
        (\{melis,thamel\}@i3s.unice.fr).}
\thanks{Soulaimane Berkane is with the Department of Computer Science and Engineering, Université du Québec en Outaouais (UQO), QC J8X 3X7, Canada (soulaimane.berkane@uqo.ca).}
\thanks{Robert Mahony is with the Systems Theory and Robotics Group, School of Engineering, Australian National University, Australia, (robert.mahony@anu.edu.au)}
\thanks{*This research is supported in part by the ASTRID ANR project ASCAR, the ``Grands Fonds Marins'' Project Deep-C, the Natural Sciences and Engineering Research Council of Canada (NSERC) under the Discovery Grant RGPIN-2020-04759, and the Fonds de recherche du Québec (FRQ).}
\thanks{\copyright~2026 IEEE. Personal use of this material is permitted.
Permission from IEEE must be obtained for all other uses, in any current
or future media, including reprinting/republishing this material for
advertising or promotional purposes, creating new collective works, for
resale or redistribution to servers or lists, or reuse of any copyrighted
component of this work in other works.}
}

\maketitle

\begin{abstract}
This paper proposes a complementary filter on $\SO(3)$ for attitude estimation from scalar measurements corresponding to projections of known inertial vectors onto body-frame sensing directions. The observer evolves directly on $\SO(3)$ and employs a constant-gain innovation tailored to the scalar-output structure. Under suitable persistence-of-excitation conditions, almost-global asymptotic stability is established when at least three inertial vectors are measured along a common body-frame direction. For configurations involving only two scalar measurements, sufficient conditions for asymptotic convergence are derived together with an explicit characterization of the region of attraction. Numerical simulations illustrate the proposed results.\end{abstract}
\section{Introduction}

Attitude estimation is a fundamental problem in navigation, guidance, and control, and has been extensively studied since the seminal formulation of Wahba's problem \cite{wahba1965least}; see \cite{crassidis2007survey} for a survey. In dynamic attitude estimation, gyroscope and vector measurements play complementary roles: gyroscopes provide angular-rate information for attitude propagation. In contrast, body-frame observations of known inertial directions provide the reference information needed to correct integration drift. Since rigid-body attitude evolves on the rotation group $\SO(3)$, nonlinear observers formulated directly on this manifold provide a natural estimation framework that preserves the rotation-matrix structure by construction.

Among geometric attitude observers, the complementary filter proposed in \cite{Mahony_Hamel_Pflimlin} constitutes a landmark contribution. The observer evolves directly on $\SO(3)$, with an innovation term constructed from discrepancies between measured body-frame directions and their predicted counterparts. It is computationally efficient and achieves almost-global asymptotic stability when at least two non-collinear inertial directions are available. Extensions to more general measurement configurations, including single time-varying inertial directions, were developed in \cite{Trumpf2012,grip2011attitude}. Further developments and stability analyses of complementary filters on $\SO(3)$ can be found in \cite{zlotnik2016nonlinear,zlotnik2016exponential,berkane2017design}. 
Despite these advances, these observers rely on complete three-dimensional vector measurements. This requirement excludes reduced-sensing configurations in which only partial vector information or scalar attitude constraints are available.

Such reduced sensing configurations may arise for different reasons. With conventional vector sensors, individual measurement channels may become unreliable because of axis-dependent noise, calibration errors, environmental disturbances, saturation, or hardware faults. In other cases, the sensing modality itself provides only scalar or low-dimensional information. For instance, cosine-type coarse Sun sensors provide a scalar voltage proportional to the projection of the Sun direction onto the sensor normal \cite{markley2014fundamentals}. In satellite navigation, individual GNSS Doppler observations provide scalar measurements of the satellite-receiver relative velocity projected along the corresponding line-of-sight direction \cite{kaplan2017understanding}. Scalar measurements can also arise from the combination of heterogeneous sensing modalities, \textit{e.g.}, landmark-based altitude constraints or tilt relations derived from barometric and range measurements \cite{alnahhal2025scalar}. Such sensing configurations motivate estimation
methods capable of exploiting incomplete, scalar, or low-dimensional
measurements directly, without first reconstructing complete vector
observations.

Early approaches to attitude estimation with reduced sensing focus on incomplete vector measurements and rely on deterministic geometric methods. The works \cite{Lee2009,Ahn2010} develop reconstruction algorithms based on partially observed reference vectors, while \cite{Shi2023} studies the orientation information that can be recovered from different combinations of three vector components. These methods operate instantaneously, without exploiting attitude kinematics or gyroscope measurements. Depending on the measurement configuration, they may yield multiple candidate attitudes or only partial orientation information. Resolving these ambiguities requires additional measurements or a selection criterion.

More recently, \cite{alnahhal2025scalar} proposed a new framework for observer design in the presence of scalar measurements. The associated measurement model accommodates not only incomplete vector observations but also measurements arising from alternative attitude constraints. Within this framework, the authors designed a globally asymptotically stable Kalman-Bucy-type filter. However, the approach relies on embedding $\mathbf{SO}(3)$ in $\mathbb{R}^9$, which leads to strong persistence of excitation requirements to guarantee uniform observability. Within the Riccati-observer framework \cite{hamel2018riccati}, an observer evolving directly on $\SO(3)$ was proposed in \cite{melis2026scalar}, achieving local exponential stability under relaxed persistence of excitation conditions, while also incorporating gyroscope bias estimation. Nevertheless, the stability guarantees remain local, and the observer still requires solving a Riccati equation.

In this work, we revisit the complementary-filter framework on $\SO(3)$ developed in \cite{Mahony_Hamel_Pflimlin,Trumpf2012} under the scalar measurement model introduced in \cite{alnahhal2025scalar}, in place of conventional full-vector measurements. Unlike the Riccati-based observers in \cite{alnahhal2025scalar,melis2026scalar}, the proposed observer evolves directly on $\SO(3)$ and uses a constant gain. Its innovation incorporates pseudoinverse normalizations that compensate for the geometry of the inertial and body-frame sensing directions and preserve the standard trace-based Lyapunov structure. Under suitable persistence-of-excitation conditions, we establish almost-global asymptotic stability when at least three inertial vectors are measured along a common body-frame direction. For two-scalar configurations (either two inertial vectors measured along a common body-frame direction or one inertial vector measured along two body-frame directions), we derive sufficient conditions for convergence from an explicitly characterized region of attraction. To the best of the authors' knowledge, these are the first stability guarantees for constant-gain complementary filtering based solely on scalar measurements.

\section{Preliminary Material}\label{section: preliminary material}
\subsection{Notation}
The Euclidean norm of a vector $x \in \mathbb{R}^n$ is denoted by $|x|$ and the 2-norm of a matrix $X\in\R^{n\times m}$ is denoted by $\|X\|$. The trace of matrix $A\in \R^{n\times n}$ is denoted by $\tr(A)$. The identity matrix is denoted by $I_n \in \mathbb{R}^{n \times n}$. For any symmetric matrix $A\in \R^{n\times n}$, $\lambda_1(A)\leq\cdots\leq\lambda_n(A)$ denote the eigenvalues of $A$ ordered increasingly. The unit $(n-1)$-sphere is represented as $\mathbb S^{n-1} := \{v \in \mathbb{R}^{n} \mid |v| = 1\}$. For any vector $v\in \mathbb S^2$, we denote by $\pi_v:=I_3-vv^\top$ the orthogonal projection onto the plane orthogonal to $v$. For any matrix $A\in\R^{m\times n}$, $\operatorname{rank}(A)$ denotes its rank, and $A^\dagger$ its Moore--Penrose pseudoinverse, defined as the unique matrix satisfying $AA^\dagger A = A,\ A^\dagger AA^\dagger=A^\dagger, \ (AA^\dagger)^\top = AA^\dagger$, and $(A^\dagger A)^\top = A^\dagger A$.
The angle between two vectors $v_1,v_2\in\R^3$ is denoted by $\angle(v_1,v_2)$. Two possibly time-varying vectors $v_1(t),v_2(t)$ are uniformly non-collinear if there exists $c>0$ such that $ |v_1\times v_2|\geq c$ for all $t\geq0$.

The special orthogonal group of 3D rotations is given by
$\SO(3) := \{R \in \mathbb{R}^{3\times3} \mid R^{\top} R = RR^\top = I_3, \det(R) = 1\}$. Its tangent space at $R\in \SO(3)$ is denoted by $T_R\SO(3)$ and
its Lie algebra is $\mathfrak{so}(3) := \{[\Omega]_\times \in \mathbb{R}^{3\times3} \mid \Omega \in \mathbb{R}^3\}$,
where $(\cdot)_\times: \mathbb{R}^3 \to \mathfrak{so}(3)$ denotes the skew-symmetric matrix operator associated with the cross product, which satisfies $[\Omega]_\times v=\Omega\times v$, for all $\Omega,v\in\mathbb{R}^3$. For any matrix $A\in\R^{3\times 3}$ and $R\in\SO(3)$, $\Pi_{T_R \SO(3)}(A)=(A-RA^\top R)/2$ denotes the orthogonal projection of $A$ onto $T_R\SO(3)$ with respect to the Frobenius inner product.
    
The following definition is used throughout the paper. 
\begin{definition}[\cite{Trumpf2012}]
    A collection of locally square-integrable vectors $v_1(t), \dots, v_\ell(t) \in \mathbb{R}^3$, with $\ell \in\mathbb{N}\setminus\{0\}$, is said to be persistently exciting if there exist $\delta,\mu>0$ such that
    \begin{align}\label{condition: PE notation}
	   \lambda_2\left(\frac1\delta \int_t^{t+\delta} \mathcal V(s)\mathcal V^\top(s)\,ds\right)\geq \mu,\quad \forall t\geq 0,
    \end{align}
    where $\mathcal V(t) = [\, v_1(t)\ \cdots\ v_\ell(t)\,]$.
\end{definition}
   
\subsection{System equations and measurements}
Consider a rigid body rotating with respect to an inertial frame $\{\mathcal{I}\}$. Let $\{\mathcal{B}\}$ be a body-fixed reference frame attached to the body and $R\in \SO(3)$ be the rotation matrix describing the orientation of frame $\{\mathcal{B}\}$ with respect to frame $\{\mathcal{I}\}$. We model the attitude kinematics of the rotating body as
\begin{align}\label{eq:kinematics}
    \dot R = R[\Omega]_\times,
\end{align}
where $\Omega\in\R^3$ is the angular velocity of $\{\mathcal{B}\}$ with respect to $\{\mathcal{I}\}$ expressed in $\{\mathcal{B}\}$, assumed to be known and typically provided by a tri-axial gyroscope. 

For the output measurements, we assume the body to be equipped with a suite of sensors providing the outputs $y_{i} \in \mathbb{R}^{n_i}$, $i=1,\cdots,p$, defined as
\begin{equation}\label{eq:general_output}
    y_i := \Lambda_i^\top R^\top b_i, \qquad \Lambda_i:=\begin{bmatrix}
        a_1^{(i)} &\cdots& a_{n_i}^{(i)}
    \end{bmatrix}\in\R^{3   \times n_i},
\end{equation}
which compose the full output measurement vector 
\begin{equation*}
    y:=(y_1^\top, \ \cdots,\ y_p^\top)\in \R^{\sum_{i=1}^pn_i}.
\end{equation*} 
Each $y_i$ represents the set of scalar outputs associated with the possibly time-varying known inertial vector $b_i \in \mathbb{R}^3$, collected by measuring the body-frame expression $R^\top b_i$ along $n_i$ ($n_i\geq 1$) body-frame vectors  ${a}_j^{(i)} \in \mathbb{R}^3$, $j = 1,\cdots, n_i$. Thus, the $j$th component of $y_i$ is
the scalar measurement
\[
    y_{ij} = (a_j^{(i)})^\top R^\top b_i.
\]
The classical full-vector measurement model is recovered by
setting $n_i=3$ and $\Lambda_i=I_3$, in which case
$
    y_i=R^\top b_i,
$
and the standard complementary-filter design applies directly
\cite{Mahony_Hamel_Pflimlin}. More generally, if
$\operatorname{rank}(\Lambda_i)=3$, the full vector $R^\top b_i$ can be
recovered from $y_i$, reducing the problem to the classical setting.

The focus of this work is the rank-deficient case
$\operatorname{rank}(\Lambda_i)<3$, which necessarily occurs when
$n_i\leq2$. In this case, $y_i$ provides only partial information
about $R^\top b_i$, and the full vector cannot be reconstructed from
$y_i$ alone. We therefore seek to construct a complementary filter
directly from these scalar measurements. Representative sensing
configurations covered by this model can be found in
\cite{alnahhal2025scalar,melis2026scalar}, with further illustrative
examples presented later.

\section{Main results}\label{section:observer design}
This section develops a complementary filter on $\SO(3)$ for scalar measurements and characterizes its convergence under different measurement configurations. Almost-global asymptotic stability is established when at least three inertial vectors are measured along a common body-frame direction under suitable persistence-of-excitation conditions. For configurations involving only two scalar measurements, sufficient conditions for asymptotic convergence are derived together with an explicit characterization of the region of attraction.

In particular, we consider a complementary observer evolving on $\SO(3)$ of the form
\begin{align}\label{eq:observer dynamics}
    \dot{\hat R} = \hat R[\Omega]_\times + [\Delta]_\times \hat R,\quad\hat R(0)\in\SO(3),
\end{align}
where $\hat R(t) \in \SO(3)$ denotes the attitude estimate at time $t$, $\Omega \in \mathbb{R}^3$ is the measured angular velocity, and $\Delta \in \mathbb{R}^3$ is an innovation term to be designed. Define the output error associated with \eqref{eq:general_output} as
\begin{equation}\label{eq:output_error_brute}
    \tilde{y}_i := \Lambda_i^\top \hat{R}^\top b_i - y_i 
\end{equation}
The proposed innovation term is defined as
\begin{align}\label{eq:innovation_brute}
    \Delta
    :=
    k\sum_{i=1}^p
    [S^\dagger b_i]_\times
    \hat R(\Lambda_i^\top)^\dagger\tilde y_i,
    \qquad
    S:=\sum_{i=1}^p b_i b_i^\top,
\end{align}
with $k>0$.
The structure of the proposed innovation is motivated by the
gradient-like observer design principle of \cite{lageman2010gradient}.
Consider first a single scalar measurement
\[
    y=a^\top R^\top b,
\]
where $a,b\in\R^3$ are known nonzero vectors, and define
\[
    \tilde y=a^\top\hat R^\top b-y.
\]
Consider the normalized squared output error
\begin{align}\label{eq:cost_function}
    f_y(\hat R)
    :=
    \frac{\tilde y^2}{|a|^2|b|^2},
    \qquad
    \nabla_{\hat R}f_y
    =
    \frac{2\tilde y}{|a|^2|b|^2}ba^\top .
\end{align}
Both expressions are invariant under arbitrary nonzero scaling of
$a$ or $b$. Following \cite{lageman2010gradient}, the Riemannian
gradient on $\SO(3)$ is obtained by orthogonally projecting the
Euclidean gradient onto $T_{\hat R}\SO(3)$:
\begin{align}\label{eq:proj_tangent_space}
    \Pi_{T_{\hat R}\SO(3)}
    (\nabla_{\hat R}f_y)
    &=
    \frac{
        \nabla_{\hat R}f_y
        -\hat R\nabla_{\hat R}f_y^\top\hat R
    }{2}
    \nonumber\\
    &=
    -\left[
        \frac{[b]_\times\hat R a}{|a|^2|b|^2}
        \tilde y
    \right]_\times\hat R .
\end{align}
Since
\[
    (a^\top)^\dagger=\frac{a}{|a|^2},
    \qquad
    (bb^\top)^\dagger b=\frac{b}{|b|^2},
\]
the innovation \eqref{eq:innovation_brute} reduces in this case to
\begin{align}\label{eq:innovation one scalar}
    \Delta
    =
    k
    \left[\frac{b}{|b|^2}\right]_\times
    \hat R\frac{a}{|a|^2}\tilde y,
    \qquad k>0.
\end{align}
Hence, the correction term $[\Delta]_\times\hat R$ is directed along
the negative Riemannian gradient of $f_y$.
The general innovation \eqref{eq:innovation_brute} extends this
single-measurement construction by using $(\Lambda_i^\top)^\dagger$
and $S^\dagger$ to account for the geometry of the body-frame sensing
vectors and inertial reference vectors, respectively. These pseudoinverse normalizations play a key role in the
Lyapunov analysis: under the conditions stated below, they ensure that
the standard trace potential on $\SO(3)$ is nonincreasing along the
attitude-error trajectories, thereby enabling the subsequent
convergence analysis.

The proposed innovation recovers the classical complementary
filter under full-vector measurements. Indeed, suppose that
$\Lambda_i=I_3$ and that the inertial vectors
$\{b_i\}_{i=1}^p$ are orthonormal. Then $S$ is the orthogonal
projector onto their span and $S^\dagger b_i=b_i$. Consequently,
\begin{align*}
    \Delta
    &=
    k\sum_{i=1}^p
    [b_i]_\times
    \hat R
    \big(\hat R^\top b_i-R^\top b_i\big)\\
    &=
    k\hat R
    \sum_{i=1}^p
    [R^\top b_i]_\times
    \hat R^\top b_i,
\end{align*}
which is the standard cross-product innovation
\cite{Mahony_Hamel_Pflimlin}.

\subsection{Stability Analysis}\label{section:observer anaylisis}
In this section, the main stability results for the proposed generalized attitude observer are presented. Define the right-invariant attitude error $\tilde R:=\hat R R^\top\in\SO(3)$ and its principal rotation angle $\tilde\theta:=\arccos\!\left(\frac{\tr(\tilde R)-1}{2}\right)\in[0,\pi]$. It follows from \eqref{eq:kinematics} and \eqref{eq:observer dynamics} that 
\begin{equation}\label{eq:closed-loop}
\dot{\tilde R}=[\Delta]_\times\tilde R.    
\end{equation}
We first consider $p\geq3$ inertial vectors whose body-frame representations are measured along the same set of $n\geq1$ body-frame vectors, collected in $\Lambda=[\,a_1\ \cdots\ a_n\,]$. The minimal configuration is $p=3$ and $n=1$, corresponding to the three scalar measurements $y_i=a_1^\top R^\top b_i$, $i=1,2,3$.

\begin{theorem}[At least three scalar measurements]
\label{theorem:3b}
Consider $p\geq3$ known inertial vectors $b_i(t)$, $i=1,\ldots,p$,
whose body-frame representations $R^\top(t)b_i(t)$ are measured along
a common set of $n\geq1$ possibly time-varying body-frame vectors,
collected in $\Lambda_i(t)=\Lambda(t)\in\R^{3\times n}$.
Assume that $S(t)$ is uniformly positive definite, i.e.,
$\lambda_1(S(t))\geq\mu_S>0$, and that $\Omega(t)$, $b_i(t)$, and $\Lambda(t)$ are uniformly continuous and bounded.
Suppose further that $\Lambda(t)$ has constant rank and that its nonzero singular values are uniformly bounded away from zero.
If $R(t)\Lambda(t)$ is persistently exciting, then
$\tilde R=I_3$ is an almost globally asymptotically stable equilibrium
of \eqref{eq:closed-loop}.
\end{theorem}
\begin{proof}
Define 
\begin{align}\label{eq: P 3b proof}
    P:=R\Lambda\Lambda^\dagger R^\top,
    \qquad
    \hat P:=\tilde R P\tilde R^\top .
\end{align}
Since $S\succeq\mu_S I_3$, one has $S^\dagger=S^{-1}$. Moreover, using
$(\Lambda^\top)^\dagger\Lambda^\top=\Lambda\Lambda^\dagger$ and
$\tilde R=\hat R R^\top$, the innovation \eqref{eq:innovation_brute} can be
written as
\[
    \Delta
    =
    k\sum_{i=1}^p
    [S^{-1}b_i]_\times
    \hat P(I_3-\tilde R)b_i .
\]
Using $[[x]_\times y]_\times=yx^\top-xy^\top$ and
$\sum_{i=1}^p b_ib_i^\top=S$, it follows that
\begin{align*}
    [\Delta]_\times
    &=
    k\Big(
    \hat P(I_3-\tilde R)
    -(I_3-\tilde R)^\top\hat P
    \Big)\\
    &=
    k(\tilde R^\top\hat P-\hat P\tilde R)
     =k(P\tilde R^\top-\tilde R P).
\end{align*}
Considering the potential function $V:=\tr(I_3-\tilde R)$, its time derivative satisfies
\begin{align}\label{eq:Vdot_P}
    \dot V
    &=-\tr([\Delta]_\times\tilde R)
      =-k\tr\!\left(P(I_3-\tilde R^2)\right).
\end{align}
Since $\Lambda\Lambda^\dagger$ is the orthogonal projector onto
$\operatorname{Im}(\Lambda)$, $P$ is the orthogonal projector onto
$\operatorname{Im}(R\Lambda)$. Let
$r:=\operatorname{rank}(\Lambda)$ and let
$\{p_1,\ldots,p_r\}$ be an orthonormal basis of
$\operatorname{Im}(R\Lambda)$. Then
\[
    P=\sum_{j=1}^r p_jp_j^\top,
\]
and \eqref{eq:Vdot_P} yields
\begin{align}
    \dot V
    &=
    -k\sum_{j=1}^r
    \left(1-p_j^\top\tilde R^2p_j\right)
    \nonumber\\
    &=
    -\frac{k}{2}
    \sum_{j=1}^r
    \left|(I_3-\tilde R^2)p_j\right|^2
    \leq0.
    \label{eq:Vdot_projector}
\end{align}

Since $\Lambda(t)$ is bounded, there exists $M>0$ such that
$\|\Lambda(t)\|\leq M$ for all $t\geq0$. Moreover, since $P$ is the
orthogonal projector onto $\operatorname{Im}(R\Lambda)$,
\[
    R\Lambda\Lambda^\top R^\top
    \preceq
    M^2 P .
\]
Hence, if $R\Lambda$ is persistently exciting, then
\[
    \lambda_2\!\left(
    \frac{1}{\delta}
    \int_t^{t+\delta}P(s)\,ds
    \right)
    \geq
    \frac{\mu}{M^2}>0 .
\]
Now, let $\underline{\sigma}>0$ be a uniform lower bound on the nonzero singular values of $\Lambda$, and set $Q:=\Lambda\Lambda^\dagger$. For any $s,t\geq0$, since $(I_3-Q(s))\Lambda(s)=0$, one has
\[
    (I_3-Q(s))Q(t)=(I_3-Q(s))(\Lambda(t)-\Lambda(s))\Lambda^\dagger(t).
\]
Together with $\|\Lambda^\dagger(t)\|\leq\underline{\sigma}^{-1}$, this yields
\[
\begin{aligned}
    \|(I_3-Q(s))Q(t)\|&\leq\frac{1}{\underline{\sigma}}
    \|\Lambda(t)-\Lambda(s)\|,\\\|(I_3-Q(t))Q(s)\|&\leq\frac{1}{\underline{\sigma}}\|\Lambda(t)-\Lambda(s)\|.
\end{aligned}
\]
Using
\[
    Q(t)-Q(s)=(I_3-Q(s))Q(t)-Q(s)(I_3-Q(t)),
\]
and the symmetry of $Q(t)$ and $Q(s)$, it follows that
\[
\|Q(t)-Q(s)\|\leq\frac{2}{\underline{\sigma}}\|\Lambda(t)-\Lambda(s)\|.
\]
Hence, uniform continuity of $\Lambda$ implies that $Q$ is uniformly continuous. Since $\Omega$ is bounded, $R$ is uniformly continuous, and therefore so is $P=RQR^\top$.

The remainder of the proof follows from the arguments of
Proposition~4.6 in \cite{Trumpf2012}, yielding almost-global
asymptotic stability of $\tilde R=I_3$.

\end{proof}
\begin{remark}
   For $p = 3$, let $b_1$, $b_2$, and $b_3$ form an orthonormal basis, and consider $\Lambda = a \in \mathbb S^2$. Then
\[
    y_i=b_i^\top Ra,\qquad i=1,2,3,
\]
and therefore
\[
    Ra=\sum_{i=1}^3 b_i y_i.
\]
Thus, the three scalar measurements are equivalent to measuring the
body-fixed direction $a$ in the inertial frame. If $Ra$ is
persistently exciting, one recovers the classical result presented
in \cite{Trumpf2012}.

When $\operatorname{rank}(\Lambda)\geq2$, persistence of excitation
is in fact not required. To see this, let
$\tilde R=\exp(\tilde\theta[u]_\times)$, with $u\in \mathbb S^2$, and set
$r:=\operatorname{rank}(\Lambda)$. Since $P$ is an orthogonal
projector of rank $r$, \eqref{eq:Vdot_P} gives
\[
    \dot V
    =
    -k(1-\cos(2\tilde\theta))(r-u^\top Pu).
\]
For $r\geq2$, $u^\top Pu\leq1$, and hence $
    \dot V
    \leq
    -2k\sin^2(\tilde\theta)$.
Therefore, $\tilde R=I_3$ is almost globally asymptotically stable
without any persistence-of-excitation requirement. This case
corresponds to inertial-frame observations of at least two independent
body-frame vectors, i.e., the frame-dual counterpart of the
classical setting involving body-frame observations of non-collinear
inertial vectors \cite{Mahony_Hamel_Pflimlin}.
\end{remark}
The following example illustrates the rank-one case of
Theorem~\ref{theorem:3b}, for which temporal excitation is required.
\begin{example}[Three-satellite GNSS Doppler and ultrasonic anemometer]
\label{example:3b}
Consider a UAV equipped with a three-dimensional ultrasonic anemometer and a GNSS receiver providing individual satellite Doppler measurements. Under negligible wind, the anemometer provides the body-frame velocity $a(t)=R^\top(t)v_I(t)$, where $v_I(t)\in\R^3$ is the UAV inertial velocity.

Suppose that Doppler measurements from three GNSS satellites, indexed by $i=1,2,3$, are available, and let $b_i(t)\in \mathbb S^2$ denote the known line-of-sight direction from the UAV to satellite $i$, obtained from the satellite ephemeris and a coarse UAV position, see Fig.~\ref{fig:gnss_configuration}. After compensating for the known satellite terms and the receiver clock drift, the resulting measurements are
\begin{align}\label{eq:output_example_3b}
    y_i=b_i^\top v_I=a^\top R^\top b_i,
\qquad i=1,2,3.
\end{align}

This configuration corresponds to $p=3$ and $\Lambda=a$. If the satellites' line-of-sight directions remain in a uniformly nondegenerate geometry, \textit{i.e.}, they do not approach a coplanar configuration, then $S(t)$ is uniformly positive definite. Since $Ra=v_I$, Theorem~\ref{theorem:3b}  guarantees almost-global asymptotic convergence provided that $v_I(t)$ is persistently exciting and bounded away from zero.
\end{example}

When only two scalar measurements are available, convergence is established on a restricted, explicitly characterized region of attraction. Two configurations are considered next: two inertial vectors measured along a common body-frame vector, and one inertial vector measured along two distinct body-frame vectors. In both cases, the size of the guaranteed region of attraction depends on the relative geometry of the inertial and body-frame sensing vectors.

\begin{lemma}[One component of two inertial vectors]\label{lemma:2b 1a}
Consider two uniformly non-collinear inertial vectors $b_1(t),b_2(t)\in\R^3$, each measured along the same possibly time-varying body-frame vector $a(t)\in\R^3$, with $|a(t)|\geq\underline{a}>0$ for all $t\geq0$. Assume that $\Omega(t)$, $b_1(t)$, $b_2(t)$, and $a(t)$ are uniformly continuous and bounded. Suppose that there exist $\theta^\star\in(0,\pi/2)$ and $\epsilon\in[0,\cos(\theta^\star/2)\cos(\theta^\star))$ such that
\begin{align}\label{condition:2b_1a_epsilon}
  \sin\!\left(\angle\!\left(a(t),R^\top(t)
\big(b_1(t)\times b_2(t)\big)\right)\right)\leq\epsilon,
\qquad \forall t\geq0.  
\end{align}
If $R(t)a(t)$ is persistently exciting, then $\tilde R=I_3$ is an asymptotically stable equilibrium of the closed-loop error dynamics, and its basin of attraction contains all initial conditions such that $\tilde{\theta}(0)\leq \theta^\star$.
\end{lemma}
The proof is reported in Appendix~\ref{proof:2b 1a}.
Lemma~\ref{lemma:2b 1a} shows that attitude estimation can be achieved using one scalar projection of each of two inertial vectors along a common body-frame vector. This relaxes the classical setting in which two complete three-dimensional vector measurements are available \cite{Trumpf2012}. We now consider the sensing configuration of Example~\ref{example:3b} with only two GNSS Doppler measurements and show that Lemma~\ref{lemma:2b 1a} still guarantees convergence.
\begin{example}[Two-satellite GNSS Doppler and ultrasonic anemometer]
\label{example:2b_1a}
Consider the UAV configuration of Example~\ref{example:3b}, using Doppler measurements from only two GNSS satellites:
\[
    y_1=a^\top R^\top b_1,
    \qquad
    y_2=a^\top R^\top b_2.
\]
Assume, without loss of generality, that $(b_1\times b_2)/|b_1\times b_2|=e_2$, and consider the yaw motion $R(t)=R_z(\psi(t))$ with $v_I(t) = |v_I|[\cos\psi(t)\ \sin\psi(t)\ 0]^\top$, where
$\psi(t)=-\pi/2+\psi_0\sin(\omega t)$, with $\psi_0\in(0,\pi/6)$, $\omega\neq0$. Then, it is not difficult to verify that
\[
\sin\!\left(\angle\!\left(
a(t), R^\top(t)(b_1\times b_2)
\right)\right)
\!=\!
\left|\sin\!\left(\psi_0\sin(\omega t)\right)\right|
\!\leq\sin(\psi_0).
\]
Moreover, $R(t)a(t) = v_I(t)$ is persistently exciting. Lemma~\ref{lemma:2b 1a} therefore applies with $\epsilon=\sin(\psi_0)$. For $\psi_0=15^\circ$, convergence is guaranteed for initial attitude-error angles strictly smaller than $71.4^\circ$.
\end{example}

\begin{lemma}[Two components of a single inertial vector]
\label{lemma:1b_2a}
Consider a possibly time-varying known inertial vector $b(t)\in\R^3$ measured along two uniformly non-collinear body-frame vectors $a_1(t),a_2(t)\in\R^3$, collected in $\Lambda(t)=[\,a_1(t)\ a_2(t)\,]$. Assume that $\Omega(t)$, $b(t)$, $a_1(t)$, and $a_2(t)$ are uniformly continuous and bounded, and that $|b(t)|\geq\underline{b}>0$ for all $t\geq0$. Suppose that there exist $\theta^\star\in(0,\pi/2)$ and $\epsilon\in[0,\cos(\theta^\star/2)\cos(\theta^\star))$ such that
\begin{align}\label{condition:1b 2a epsilon}
\sin\!\left(\angle\!\left(
a_1(t)\times a_2(t),R^\top(t)b(t)
\right)\right)\leq\epsilon,
\qquad \forall t\geq0.
\end{align}
If $b(t)$ is persistently exciting, then $\tilde R=I_3$ is an asymptotically stable equilibrium of the closed-loop error dynamics \eqref{eq:closed-loop}, and its basin of attraction contains all initial conditions such that $\tilde \theta(0)\leq \theta^\star$.
\end{lemma}

The proof is reported in Appendix~\ref{proof: 1b 2a}. Lemma~\ref{lemma:1b_2a} shows that attitude estimation can be achieved
using only two scalar projections of the body-frame representation of
a single time-varying inertial vector. This relaxes the classical
single-vector setting in which the complete three-dimensional vector
measurement is available \cite{Trumpf2012}. We now provide an
illustrative application.
\begin{example}[Two-beam Doppler radar with GNSS velocity]
\label{example:1b_2a}
Consider a UAV equipped with two Doppler radar beams directed obliquely toward approximately stationary flat terrain. The two scalar measurements are collected as
\[
    y_1=
    \begin{bmatrix}a_1 & a_2\end{bmatrix}^{\!\top}
    R^\top v_I(t),
\]
where $a_1,a_2\in \mathbb S^2$ are the beam directions and $v_I(t)\in\R^3$ is the inertial velocity provided by GNSS.

Let $e_1$, $e_2$, and $e_3$ denote the forward, right, and downward body axes, respectively. Assume that the beams are symmetrically mounted about a direction tilted by $\gamma\in(0,\pi/2)$ in the $(e_1,e_3)$-plane:
\begin{align*}
a_1&=\cos(\phi)\big(\cos(\gamma)e_1+\sin(\gamma)e_3\big)
       +\sin(\phi)e_2,\\
a_2&=\cos(\phi)\big(\cos(\gamma)e_1+\sin(\gamma)e_3\big)
       -\sin(\phi)e_2,
\end{align*}
where $\phi\in(0,\pi/2)$ is the lateral spread of the beams; see Fig.~\ref{fig:doppler_configuration}. Let $\alpha,\beta\in(-\pi/2,\pi/2)$ denote the angle of attack and sideslip angle, respectively. Then
\[
R^\top v_I
=
|v_I|
\begin{bmatrix}
\cos(\alpha)\cos(\beta)\\
\sin(\beta)\\
\sin(\alpha)\cos(\beta)
\end{bmatrix},
\]
and therefore
\[
\left|
\frac{(a_1\times a_2)^\top R^\top v_I}
{|a_1\times a_2|\,|v_I|}
\right|
=
|\cos(\beta)\sin(\gamma-\alpha)|.
\]
If
\[
|\alpha(t)|\leq\bar{\alpha}
<\gamma<\frac{\pi}{2}-\bar{\alpha},
\qquad
|\beta(t)|\leq\bar{\beta}<\frac{\pi}{2},
\]
it follows that
\begin{align}\label{bound:pitot}
\sin\!\left(
\angle\!\left(
a_1\times a_2,R^\top v_I
\right)\right)
\leq
\sqrt{
1-\cos^2(\bar{\beta})
\sin^2(\gamma-\bar{\alpha})
}.
\end{align}
Consequently, the right-hand side of \eqref{bound:pitot} provides a
valid choice of $\epsilon$ in Lemma~\ref{lemma:1b_2a} and determines
the corresponding admissible range of $\theta^\star$.
If $v_I(t)$ is persistently exciting and bounded away from zero,
Lemma~\ref{lemma:1b_2a} guarantees convergence for initial
attitude-error angles not exceeding any admissible
$\theta^\star$.\end{example}

\begin{figure}[t]
    \captionsetup{font=normalsize}
    \centering

    \begin{minipage}[c]{0.49\columnwidth}
        \centering

        \resizebox{\linewidth}{!}{%
        \begin{tikzpicture}[
            beamdir/.style={
                -{Stealth[length=6.2mm,width=4.0mm]},
                blue,
                line width=2.8pt
            },
            beamlabel/.style={
                font=\Huge,
                scale=1.45,
                text=blue,
                inner sep=1pt
            }
        ]

        \node[inner sep=0pt] (drone) at (0,0)
        {
            \includegraphics[
                width=5.2cm,
                trim=30bp 272bp 80bp 169bp,
                clip
            ]{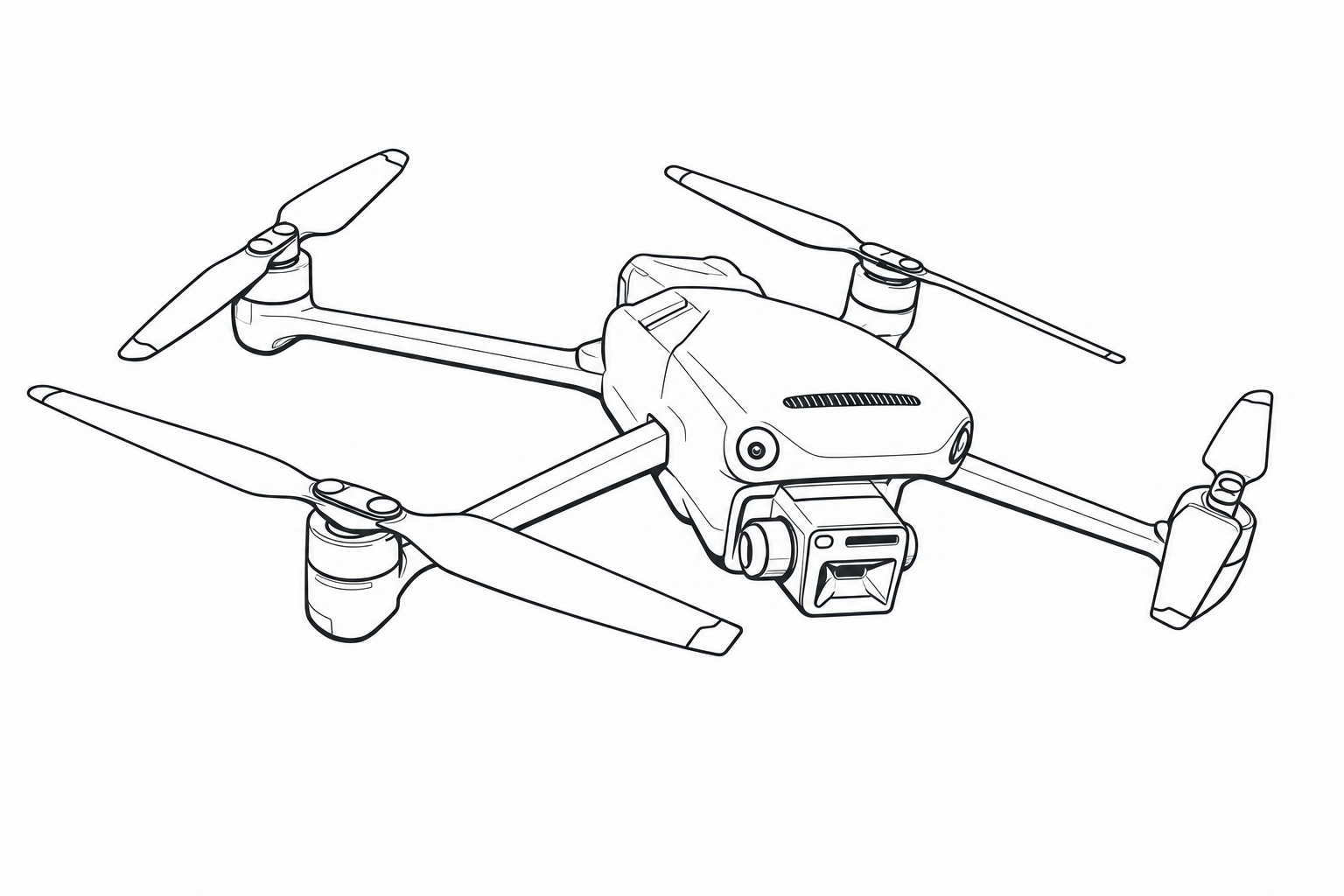}
        };

        \node[
            inner sep=0pt,
            rotate=-120
        ] (sat1) at (-4.0,4.7)
        {
            \includegraphics[
                width=4.0cm,
                trim=13bp 103bp 31bp 156bp,
                clip
            ]{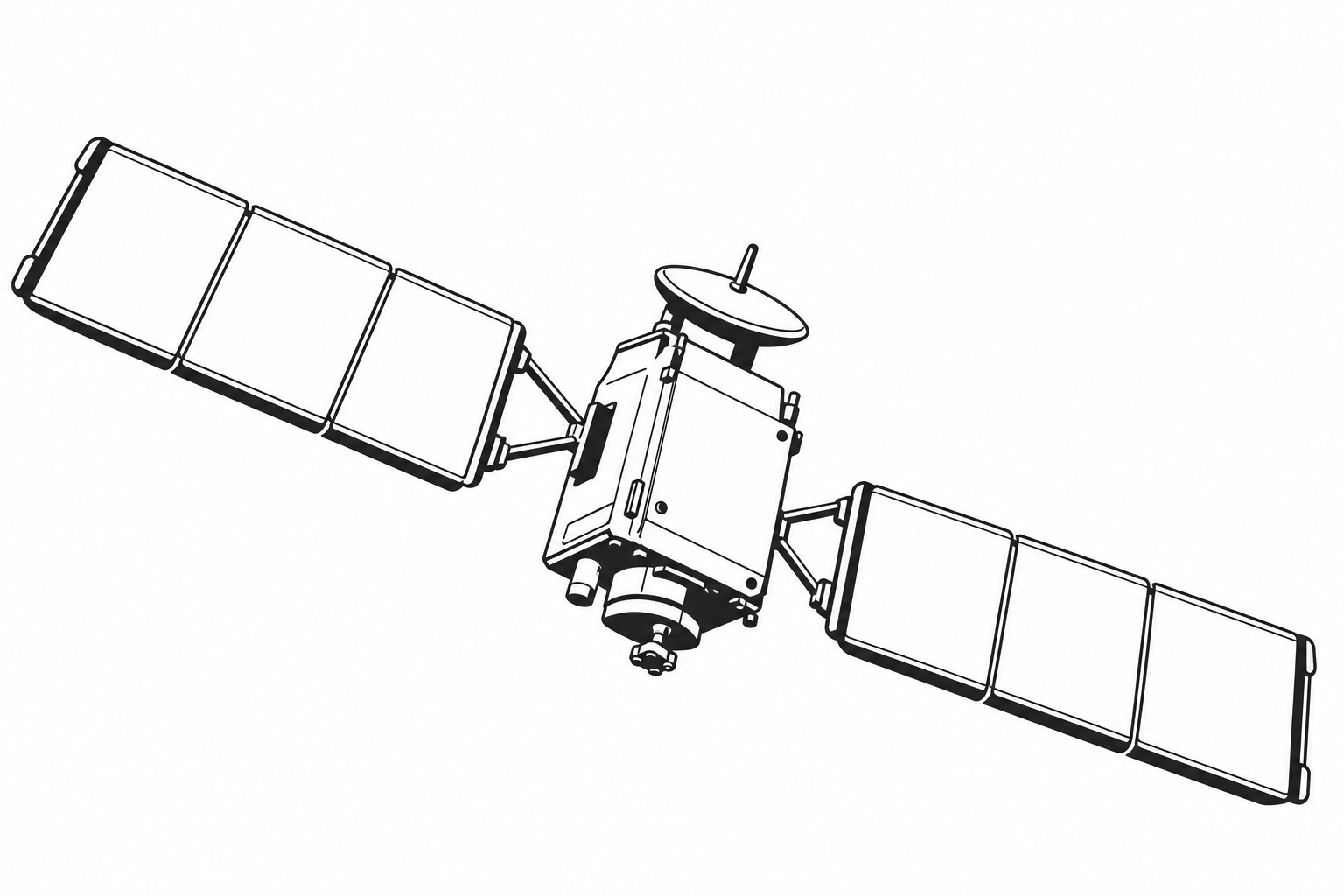}
        };

        \node[
            inner sep=0pt,
            rotate=200
        ] (sat2) at (0,5.5)
        {
            \includegraphics[
                width=4.0cm,
                trim=13bp 103bp 31bp 156bp,
                clip
            ]{images/Satellite.png}
        };

        \node[
            inner sep=0pt,
            rotate=140
        ] (sat3) at (4.0,4.7)
        {
            \includegraphics[
                width=4.0cm,
                trim=13bp 103bp 31bp 156bp,
                clip
            ]{images/Satellite.png}
        };

        \coordinate (receiver)
            at ($(drone.center)+(0,0.45)$);

        \coordinate (waveStartL)
            at ($(sat1.north)+(-0.45,+0.05)$);

        \coordinate (waveEndL)
            at ($(receiver)+(0,0.35)$);

        \coordinate (dirEndL)
            at ($(waveStartL)!0.30!(waveEndL)$);

        \draw[beamdir]
            (waveStartL) -- (dirEndL);

        \node[
            beamlabel
        ] at ($(waveStartL)!0.58!(dirEndL)+(-0.65,-0.20)$)
        {
            $b_1$
        };

        \path let
            \p1=(waveStartL),
            \p2=(waveEndL),
            \n1={atan2(\y2-\y1,\x2-\x1)+90}
        in
            \pgfextra{\xdef\rotLeft{\n1}};

        \foreach \pos/\rad in {
            0.05/0.11,
            0.14/0.14,
            0.23/0.17,
            0.32/0.20,
            0.41/0.23,
            0.50/0.26,
            0.59/0.29,
            0.68/0.32,
            0.77/0.35,
            0.86/0.38
        }
        {
            \coordinate (wL)
                at ($(waveStartL)!\pos!(waveEndL)$);

            \begin{scope}[
                shift={(wL)},
                rotate=\rotLeft
            ]
                \draw[line width=0.8pt]
                    ({\rad*cos(200)},{\rad*sin(200)})
                    arc[
                        start angle=200,
                        end angle=340,
                        radius=\rad
                    ];
            \end{scope}
        }

        \coordinate (waveStartC)
            at ($(sat2.north)+(-0.35,+0.25)$);

        \coordinate (waveEndC)
            at ($(receiver)+(0,0.35)$);

        \coordinate (dirEndC)
            at ($(waveStartC)!0.30!(waveEndC)$);

        \draw[beamdir]
            (waveStartC) -- (dirEndC);

        \node[
            beamlabel
        ] at ($(waveStartC)!0.58!(dirEndC)+(0.80,-0.02)$)
        {
            $b_2$
        };

        \path let
            \p1=(waveStartC),
            \p2=(waveEndC),
            \n1={atan2(\y2-\y1,\x2-\x1)+90}
        in
            \pgfextra{\xdef\rotCenter{\n1}};

        \foreach \pos/\rad in {
            0.05/0.11,
            0.14/0.14,
            0.23/0.17,
            0.32/0.20,
            0.41/0.23,
            0.50/0.26,
            0.59/0.29,
            0.68/0.32,
            0.77/0.35,
            0.86/0.38
        }
        {
            \coordinate (wC)
                at ($(waveStartC)!\pos!(waveEndC)$);

            \begin{scope}[
                shift={(wC)},
                rotate=\rotCenter
            ]
                \draw[line width=0.8pt]
                    ({\rad*cos(200)},{\rad*sin(200)})
                    arc[
                        start angle=200,
                        end angle=340,
                        radius=\rad
                    ];
            \end{scope}
        }

        \coordinate (waveStartR)
            at ($(sat3.north)+(0.0,+0.4)$);

        \coordinate (waveEndR)
            at ($(receiver)+(0,0.35)$);

        \coordinate (dirEndR)
            at ($(waveStartR)!0.30!(waveEndR)$);

        \draw[beamdir]
            (waveStartR) -- (dirEndR);

        \node[
            beamlabel
        ] at ($(waveStartR)!0.58!(dirEndR)+(0.80,-0.23)$)
        {
            $b_3$
        };

        \path let
            \p1=(waveStartR),
            \p2=(waveEndR),
            \n1={atan2(\y2-\y1,\x2-\x1)+90}
        in
            \pgfextra{\xdef\rotRight{\n1}};

        \foreach \pos/\rad in {
            0.05/0.11,
            0.14/0.14,
            0.23/0.17,
            0.32/0.20,
            0.41/0.23,
            0.50/0.26,
            0.59/0.29,
            0.68/0.32,
            0.77/0.35,
            0.86/0.38
        }
        {
            \coordinate (wR)
                at ($(waveStartR)!\pos!(waveEndR)$);

            \begin{scope}[
                shift={(wR)},
                rotate=\rotRight
            ]
                \draw[line width=0.8pt]
                    ({\rad*cos(200)},{\rad*sin(200)})
                    arc[
                        start angle=200,
                        end angle=340,
                        radius=\rad
                    ];
            \end{scope}
        }

        \end{tikzpicture}%
        }

        \vspace{0.5mm}
        \normalsize (a)

    \end{minipage}%
    \hfill%
    \begin{minipage}[c]{0.49\columnwidth}
        \centering

        \resizebox{\linewidth}{!}{%
        \begin{tikzpicture}[
            >=Stealth,
            axis/.style={
                -{Stealth[length=1.8mm,width=1.2mm]},
                line width=0.75pt
            },
            beamaxis/.style={
                -{Stealth[length=3.2mm,width=2.1mm]},
                line width=1.2pt,
                red
            },
            e1axis/.style={
                -{Stealth[length=3.2mm,width=2.1mm]},
                line width=1.2pt
            },
            wave/.style={
                line width=0.6pt
            },
            ground/.style={
                line width=0.8pt
            },
            framelabel/.style={
                font=\fontsize{18}{16}\selectfont,
                inner sep=1pt
            },
            anglabel/.style={
                font=\fontsize{18}{16}\selectfont,
                inner sep=1pt
            },
            beamtext/.style={
                font=\fontsize{18}{16}\selectfont,
                text=red,
                inner sep=1pt
            }
        ]

        \def\ygnd{-2.20}
        \def\xright{3.20}

        \begin{scope}[shift={(-1.35,0.35)}]

        \draw[ground]
            (-1.8,\ygnd) -- (\xright,\ygnd);

        \node[
            inner sep=0pt
        ] (img) at (0,0)
        {
            \includegraphics[
                width=6cm
            ]{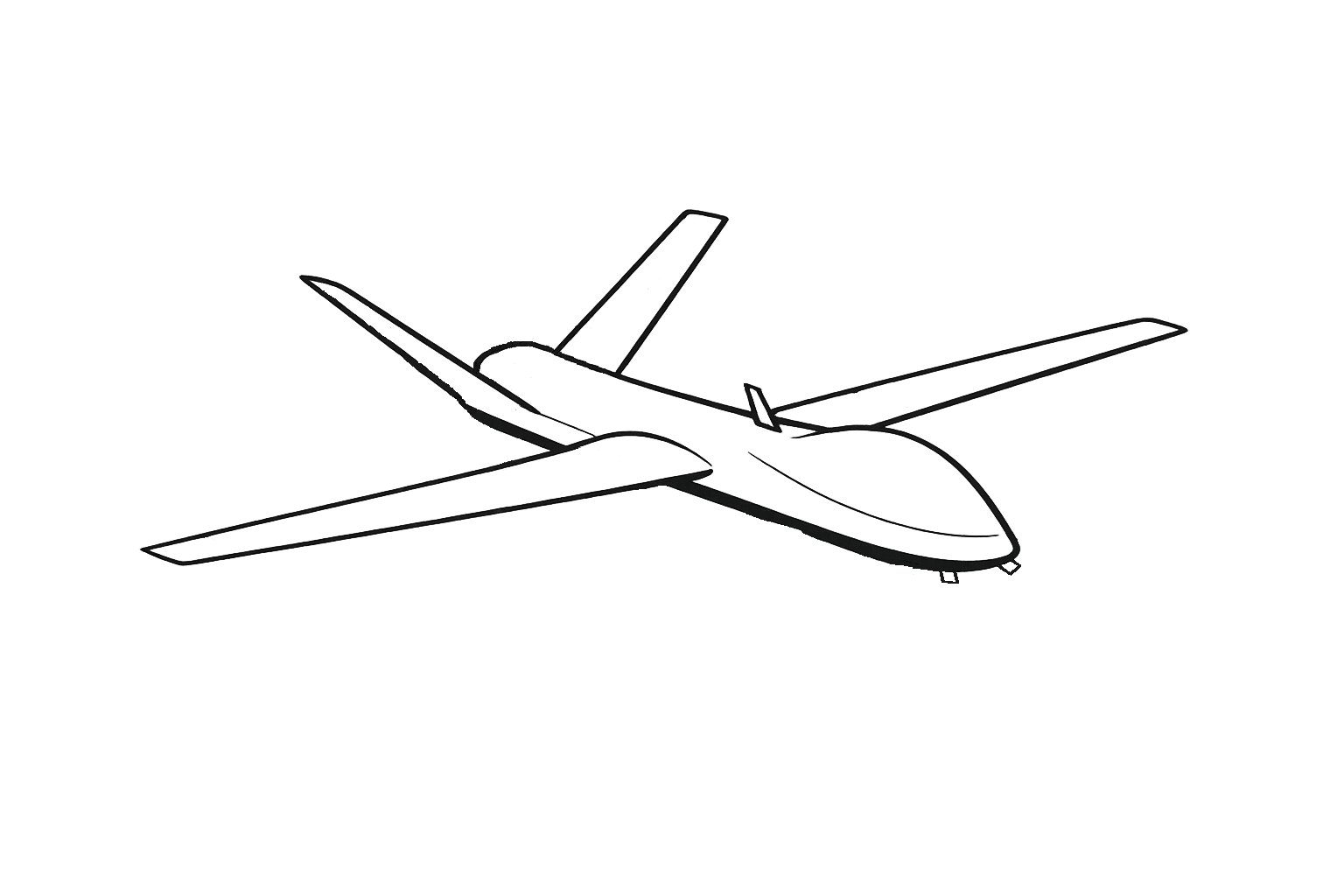}
        };

        \node[
            framelabel,
            anchor=north
        ] at ($(img.south)+(1.9,1.95)$)
        {
            $e_1$
        };

        \coordinate (O) at (1.21,-0.30);

        \coordinate (Bright) at (2.15,-0.650);
        \coordinate (Bdash)  at (1.86,-1.20);

        \coordinate (S1) at (1.24,-0.60);
        \coordinate (S2) at (1.53,-0.54);

        \coordinate (aone) at (1.26,-1.40);
        \coordinate (atwo) at (2.21,-1.04);

        \draw[e1axis]
            (O) -- (Bright);

        \draw[
            dashed,
            line width=0.6pt
        ]
            (O) -- (Bdash);

        \draw[
            line width=1.2pt
        ]
            (1.26,-1.14)
            arc[
                start angle=-95,
                end angle=-40,
                radius=0.50cm
            ];

        \node[
            anglabel
        ] at (1.55,-1.39)
        {
            $\phi$
        };

        \coordinate (E1) at (1.30,\ygnd);
        \coordinate (E2) at (\xright,-1.70);

        \path let
            \p1=(S1),
            \p2=(E1),
            \n1={atan2(\y2-\y1,\x2-\x1)+90}
        in
            \pgfextra{\xdef\rotBeamOne{\n1}};

        \coordinate (waveStart1)
            at ($(S1)!0.12!(E1)$);

        \coordinate (waveEnd1)
            at ($(S1)!0.92!(E1)$);

        \foreach \t/\r in {
            0.00/0.06,
            0.18/0.09,
            0.38/0.12,
            0.60/0.15,
            0.82/0.18
        }
        {
            \coordinate (w1)
                at ($(waveStart1)!\t!(waveEnd1)$);

            \begin{scope}[
                shift={(w1)},
                rotate=\rotBeamOne
            ]
                \draw[wave]
                    ({\r*cos(200)},{\r*sin(200)})
                    arc[
                        start angle=200,
                        end angle=340,
                        radius=\r
                    ];
            \end{scope}
        }

        \path let
            \p1=(S2),
            \p2=(E2),
            \n1={atan2(\y2-\y1,\x2-\x1)+90}
        in
            \pgfextra{\xdef\rotBeamTwo{\n1}};

        \coordinate (waveStart2)
            at ($(S2)!0.12!(E2)$);

        \coordinate (waveEnd2)
            at ($(S2)!0.92!(E2)$);

        \foreach \t/\r in {
            0.00/0.06,
            0.18/0.09,
            0.38/0.12,
            0.60/0.15,
            0.82/0.18
        }
        {
            \coordinate (w2)
                at ($(waveStart2)!\t!(waveEnd2)$);

            \begin{scope}[
                shift={(w2)},
                rotate=\rotBeamTwo
            ]
                \draw[wave]
                    ({\r*cos(200)},{\r*sin(200)})
                    arc[
                        start angle=200,
                        end angle=340,
                        radius=\r
                    ];
            \end{scope}
        }

        \draw[beamaxis]
            (S1) -- (aone);

        \draw[beamaxis]
            (S2) -- (atwo);

        \node[
            beamtext
        ] at ($(S1)!0.55!(aone)+(-0.38,-0.11)$)
        {
            $a_1$
        };

        \node[
            beamtext
        ] at ($(S2)!0.58!(atwo)+(0.52,0.12)$)
        {
            $a_2$
        };

        \draw[
            line width=1.2pt
        ]
            (2.05,-0.67)
            arc[
                start angle=-10,
                end angle=-40.5,
                radius=1.0cm
            ];

        \node[
            anglabel
        ] at (2.04,-1.25)
        {
            $\gamma$
        };

        \end{scope}

        \end{tikzpicture}%
        }

        \vspace{0.5mm}
        \normalsize (b)

    \end{minipage}

    \caption{Illustrative sensing configurations:
(a) a UAV receiving Doppler measurements from three GNSS satellites;
(b) a UAV equipped with two Doppler radar beams symmetrically arranged about a direction tilted by $\gamma$, with an angular separation of $2\phi$.}
    \label{fig:sensing_configurations}

    \makeatletter
    \begingroup
        \def\@currentlabel{\thefigure(a)}
        \label{fig:gnss_configuration}
    \endgroup
    \begingroup
        \def\@currentlabel{\thefigure(b)}
        \label{fig:doppler_configuration}
    \endgroup
    \begingroup
        \def\@currentlabel{\thefigure(b)}
        \label{fig:example_1b_2a}
    \endgroup
    \makeatother

\end{figure}

\begin{remark}
    In Lemma~\ref{lemma:2b 1a} and Lemma~\ref{lemma:1b_2a}, the size of the guaranteed region of attraction depends on the relative orientation
between the subspace spanned by the body-frame representations of the
inertial vectors and the subspace spanned by the body-frame sensing
vectors, as characterized by
\eqref{condition:2b_1a_epsilon} and
\eqref{condition:1b 2a epsilon}.

In particular, when these subspaces are orthogonal, one has
$\epsilon=0$, and any $\theta^\star<\pi/2$ satisfies the corresponding
condition. Hence, the guaranteed region of attraction can include any
initial attitude-error angle strictly smaller than $\pi/2$.
This occurs when $a$ is collinear with
$R^\top(b_1\times b_2)$ in Lemma~\ref{lemma:2b 1a}, and when
$R^\top b$ is collinear with $a_1\times a_2$ in
Lemma~\ref{lemma:1b_2a}.
\end{remark}

\section{Simulation Results}\label{section:simulation results}

This section presents simulation results illustrating the convergence
properties of the observer developed in
Section~\ref{section:observer design}. Two scenarios are considered.
The first corresponds to Examples~\ref{example:3b}
and~\ref{example:2b_1a}, using respectively three and two scalar GNSS
Doppler measurements, while the second corresponds to
Example~\ref{example:1b_2a}. For reference, the results are compared
with those obtained using the complementary filter of
\cite{Mahony_Hamel_Pflimlin} with full-vector measurements.

The first simulation considers the configurations of
Examples~\ref{example:3b} and~\ref{example:2b_1a}. The satellite--UAV
line-of-sight directions are chosen as
$b_1=
\begin{bmatrix}
\cos(30^\circ)&0&-\sin(30^\circ)
\end{bmatrix}^\top$,
$b_2=
\begin{bmatrix}
\cos(70^\circ)&0&-\sin(70^\circ)
\end{bmatrix}^\top,$ 
and $
b_3=
\begin{bmatrix}
0&\cos(45^\circ)&-\sin(45^\circ)
\end{bmatrix}^\top.$
Thus,
\[
\frac{b_1\times b_2}{|b_1\times b_2|}=e_2,
\]
while $b_1,b_2,b_3$ are linearly independent.

The UAV follows the trajectory of Example~\ref{example:2b_1a}, with
$\psi_0=15^\circ$, $\omega=1~\mathrm{rad/s}$, and
$|v_I|=15~\mathrm{m/s}$. The condition of
Lemma~\ref{lemma:2b 1a} then permits any
$\theta^\star<71.4^\circ$.

To illustrate the dependence of the scalar observers on temporal
excitation, the trajectory is held constant over
$t\in[5,35]~\mathrm{s}$ and subsequently resumed continuously.
During this interval, both the attitude and inertial velocity are
constant.

For comparison, two full-vector complementary filters are also
considered. They are assumed to have access, respectively, to the
full body-frame measurements
$\{R^\top b_1,R^\top b_2\}$ and
$\{R^\top b_1,R^\top b_2,R^\top b_3\}$.
All observers are initialized with
\[
\hat R(0)
=
R_z(-30^\circ)R_y(-45^\circ)R_x(-22.5^\circ),
\]
which yields
$\tilde\theta(0)\approx69.4^\circ$.
The innovation gain is set to $k=1$ for all observers.

The second simulation corresponds to Example~\ref{example:1b_2a}.
The Doppler-beam geometry is defined by
$\gamma=45^\circ$ and $\phi=15^\circ$, while
\[
\alpha(t)=\bar{\alpha}\sin(\omega_\alpha t),
\qquad
\beta(t)=\bar{\beta}\sin(\omega_\beta t),
\]
with
\[
\bar \alpha=20^\circ,\quad
\bar\beta=25^\circ,\quad
\omega_\alpha=0.17~\mathrm{rad/s},\quad
\omega_\beta=0.23~\mathrm{rad/s}.
\]
The condition of Lemma~\ref{lemma:1b_2a} then permits any
$\theta^\star<20.23^\circ$.

The aircraft follows the loitering trajectory
\[
v_I(t)
=
|v_I|
\begin{bmatrix}
\cos(\omega t)&\sin(\omega t)&0
\end{bmatrix}^\top,
\]
with $|v_I|=1~\mathrm{m/s}$ and
$\omega=0.35~\mathrm{rad/s}$, which ensures persistence of excitation.
The corresponding attitude is
\[
R(t)
=
R_z\big(\omega t-\beta(t)\big)
R_y\big(\alpha(t)\big).
\]

For comparison, the full-vector complementary filter is assumed to
have access to the complete body-frame measurement $R^\top v_I$.
Both observers are initialized with
\[
\hat R(0)
=
R_z(15^\circ)R_y(10^\circ)R_x(7.5^\circ),
\]
yielding $\tilde\theta(0)\approx19^\circ$.
The innovation gain is set to $k=1$ for both observers.

\subsection{Results and discussion}

The attitude-error trajectories for the two simulations are reported
in Fig.~\ref{fig:attitude error 3b} and
Fig.~\ref{fig:attitude error 1b-2a}.

In Fig.~\ref{fig:attitude error 3b}, both scalar observers converge
while the trajectory is time-varying. Their convergence essentially
stalls over $t\in[5,35]~\mathrm{s}$, during which the attitude and
inertial velocity are held constant, and resumes once temporal
excitation is restored. In contrast, the full-vector complementary
filters continue to converge because the available non-collinear
vector measurements provide instantaneous attitude information.
These results illustrate the role of temporal excitation in the
scalar-measurement configurations of Theorem~\ref{theorem:3b} and
Lemma~\ref{lemma:2b 1a}.

Figure~\ref{fig:attitude error 1b-2a} shows convergence of the
two-scalar observer of Lemma~\ref{lemma:1b_2a} from an initial error
inside the guaranteed region of attraction. The full-vector
complementary filter also converges, with access to the complete
body-frame vector $R^\top v_I$.

Since the scalar and full-vector observers exploit different
measurement information and geometry-dependent normalizations, equal
values of the gain $k$ do not imply matched convergence rates. The
comparisons are therefore intended to illustrate the convergence
properties and sensing requirements of the proposed observer rather
than to provide a quantitative comparison of convergence speed.

\begin{figure}[t]
    \centering   
    \vspace{1mm}\includegraphics[width=1.\linewidth]{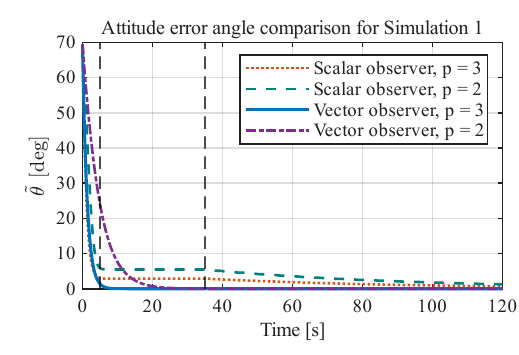}
    \caption{Attitude error $\tilde\theta$ in Simulation 1 for the proposed observer in the $3b$-$1a$ and $2b$-$1a$ configurations and the corresponding full-vector complementary filters.}
    \label{fig:attitude error 3b}
\end{figure}
\begin{figure}[t]
    \centering  
    \vspace{1mm}\includegraphics[width=1.\linewidth]{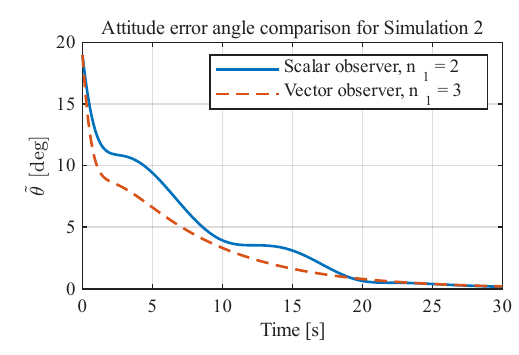}
    \caption{Attitude error $\tilde\theta$ in Simulation 2 for the proposed observer in the $1b$-$2a$ configuration and the corresponding full-vector complementary filter.}
    \label{fig:attitude error 1b-2a}
\end{figure}

\section{Conclusion}\label{section:conclusion}
This paper proposes a complementary filter on $\SO(3)$ for attitude
estimation from scalar measurements, using a constant-gain innovation
that accounts for the geometry of the sensing directions.

Three main results are established. First, almost-global asymptotic
stability is obtained when at least three known inertial vectors are
measured along a common set of body-frame vectors under a suitable
persistence-of-excitation condition, with three scalar measurements as
the minimal configuration (Theorem~\ref{theorem:3b}). Second, for two
inertial vectors measured along a common body-frame vector,
sufficient conditions for asymptotic convergence and an explicit
region of attraction are derived (Lemma~\ref{lemma:2b 1a}). Third, a
dual result is obtained for one inertial vector measured along two
independent body-frame vectors (Lemma~\ref{lemma:1b_2a}).

These results show that attitude estimation can be achieved directly
from incomplete scalar measurements by exploiting temporal excitation
and sensing geometry.
\section{Appendix}

\subsection{Proof of Lemma~\ref{lemma:2b 1a}}
\label{proof:2b 1a}
Define
\[
    \mathring a:=\frac{a}{|a|},
    \qquad
    \nu_b:=\frac{b_1\times b_2}{|b_1\times b_2|},
\]
and let
\[
    q:=R\mathring a,
    \qquad
    \hat q:=\hat R\mathring a=\tilde Rq .
\]
Since $b_1$ and $b_2$ are uniformly non-collinear,
$S=b_1b_1^\top+b_2b_2^\top$ has rank two and
\[
    S^\dagger S=\pi_{\nu_b}:=I_3-\nu_b\nu_b^\top .
\]
Moreover,
\[
    (a^\top)^\dagger a^\top
    =
    \frac{aa^\top}{|a|^2}
    =
    \mathring a\mathring a^\top .
\]
Using
\[
    \tilde y_i
    =
    a^\top\hat R^\top(I_3-\tilde R)b_i,
\]
the innovation term \eqref{eq:innovation_brute} can therefore be
rewritten as
\begin{align*}
    \Delta
    &=
    k\sum_{i=1}^2
    [S^\dagger b_i]_\times
    \hat R\mathring a\mathring a^\top
    \hat R^\top(I_3-\tilde R)b_i\\
    &=
    k\sum_{i=1}^2
    [S^\dagger b_i]_\times
    \hat q\, b_i^\top(\hat q-q)\\
    &=
    -k[\hat q]_\times
    S^\dagger\sum_{i=1}^2b_ib_i^\top(\hat q-q)\\
    &=
    k[\pi_{\nu_b}(\hat q-q)]_\times\hat q .
\end{align*}

Consider the Lyapunov function
\[
    V:=\tr(I_3-\tilde R).
\]
Using the identity
\[
    [[x]_\times y]_\times=yx^\top-xy^\top ,
\]
one obtains
\begin{align}
    \dot V
    &=
    -\tr([\Delta]_\times\tilde R)\nonumber\\
    &=
    -k(\hat q-q)^\top
    \pi_{\nu_b}
    (\tilde R^2q-q).
    \label{eq:Vdot_2b1a_initial}
\end{align}

Define
\[
    E:=qq^\top-\nu_b\nu_b^\top .
\]
Since $q$ and $\nu_b$ are unit vectors,
\begin{align}
    \|E\|
    &=
    \left|\sin\!\left(\angle(q,\nu_b)\right)\right|\nonumber\\
    &=
    \left|\sin\!\left(
    \angle(\mathring a,R^\top\nu_b)
    \right)\right|
    \leq\epsilon,
    \label{eq:E_bound_2b1a}
\end{align}
where the last inequality follows from
\eqref{condition:2b_1a_epsilon}.
Moreover,
\[
    \pi_{\nu_b}=\pi_q+E,
    \qquad
    \pi_q:=I_3-qq^\top .
\]
Hence,
\[
    \dot V=V_0+V_E,
\]
where
\begin{subequations}\label{eq:V0_VE_Lemma_1}
\begin{align}
    V_0
    &:=
    -k(\hat q-q)^\top
    \pi_q(\tilde R^2q-q),\\
    V_E
    &:=
    -k(\hat q-q)^\top
    E(\tilde R^2q-q).
\end{align}
\end{subequations}

Let
\[
    \tilde R
    =
    I_3+\sin(\tilde\theta)[u]_\times
    +(1-\cos(\tilde\theta))[u]_\times^2,
\]
where $u\in\mathbb S^2$ is the rotation axis of $\tilde R$.
Straightforward computations give
\begin{align}
    V_0
    &=
    -k(q^\top\hat q)
    \left(1-q^\top\tilde R^2q\right).
    \label{eq:V0_exact}
\end{align}
Furthermore,
\[
    q^\top\hat q
    =
    q^\top\tilde Rq
    =
    \cos(\tilde\theta)
    +(1-\cos(\tilde\theta))(u^\top q)^2
    \geq\cos(\tilde\theta).
\]
Therefore, as long as $\tilde\theta<\pi/2$,
\begin{align}
    V_0
    \leq
    -k\cos(\tilde\theta)
    \left(1-q^\top\tilde R^2q\right).
    \label{eq:V0_bound}
\end{align}

For the perturbation term, using \eqref{eq:E_bound_2b1a},
\begin{align*}
    |V_E|
    &\leq
    k\epsilon
    |\hat q-q|\,
    |\tilde R^2q-q|.
\end{align*}
The following identities hold:
\begin{align}
    |\tilde R^2q-q|^2
    &=
    2\left(1-q^\top\tilde R^2q\right),
    \label{eq:norm_R2}\\
    |\hat q-q|^2
    &=
    2(1-q^\top\tilde Rq)\nonumber\\
    &=
    2(1-\cos\tilde\theta)
    \left(1-(u^\top q)^2\right)\nonumber\\
    &=
    \frac{
    1-q^\top\tilde R^2q
    }{
    2\cos^2(\tilde\theta/2)
    }.
    \label{eq:norm_R1}
\end{align}
Consequently,
\begin{align}
    |V_E|
    \leq
    k\frac{\epsilon}{\cos(\tilde\theta/2)}
    \left(1-q^\top\tilde R^2q\right).
    \label{eq:VE_bound}
\end{align}

Combining \eqref{eq:V0_bound} and \eqref{eq:VE_bound} yields
\begin{align}
    \dot V
    \leq
    -k
    \left(1-q^\top\tilde R^2q\right)
    \left(
    \cos(\tilde\theta)
    -
    \frac{\epsilon}{\cos(\tilde\theta/2)}
    \right).
    \label{eq:Vdot_final_2b1a}
\end{align}

Now consider an initial condition satisfying
\[
    \tr(\tilde R(0))
    \geq
    1+2\cos(\theta^\star),
\]
which is equivalent to
$\tilde\theta(0)\leq\theta^\star$.
For
$\tilde\theta\in[0,\theta^\star]$, one has
\[
    \cos(\tilde\theta)
    -
    \frac{\epsilon}{\cos(\tilde\theta/2)}
    \geq
    \cos(\theta^\star)
    -
    \frac{\epsilon}{\cos(\theta^\star/2)}
    =:c^\star.
\]
By assumption,
\[
    \epsilon
    <
    \cos(\theta^\star/2)\cos(\theta^\star),
\]
and hence $c^\star>0$. Therefore,
\begin{align}
    \dot V
    \leq
    -kc^\star
    \left(1-q^\top\tilde R^2q\right)
    \leq0.
    \label{eq:Vdot_negative_2b1a}
\end{align}

Since
\[
    V=2(1-\cos\tilde\theta),
\]
and $V$ is nonincreasing, the set
\[
    \{\tilde R\in\SO(3):
    \tilde\theta\leq\theta^\star\}
\]
is forward invariant. Moreover, integrating
\eqref{eq:Vdot_negative_2b1a} gives
\[
    1-q^\top\tilde R^2q\in L_1.
\]
The boundedness and uniform continuity assumptions imply that
$1-q^\top\tilde R^2q$ is uniformly continuous. Hence, by Barbalat's
lemma,
\[
    1-q^\top\tilde R^2q\longrightarrow0.
\]
Using \eqref{eq:norm_R1} and
$\tilde\theta\leq\theta^\star<\pi/2$, it follows that
\[
    |\hat q-q|\longrightarrow0,
\]
that is,
\[
    |(\tilde R-I_3)R\mathring a|\longrightarrow0.
\]
The expression of $\Delta$ then also gives
$\Delta\to0$, and consequently
$\dot{\tilde R}\to0$.

Finally, since $a$ is bounded, there exists $\bar a>0$ such that
$|a(t)|\leq \bar a$. Since
\[
    Ra=|a|R\mathring a,
\]
one has
\[
    (R\mathring a)(R\mathring a)^\top
    =
    \frac{(Ra)(Ra)^\top}{|a|^2}
    \succeq
    \frac{1}{\bar{a}^2}(Ra)(Ra)^\top.
\]
Thus, persistence of excitation of $Ra$ implies persistence of
excitation of $R\mathring a$. The standard persistence-of-excitation
argument of Proposition~4.6 in \cite{Trumpf2012}, applied together
with
\[
    |(\tilde R-I_3)R\mathring a|\to0,
    \qquad
    \dot{\tilde R}\to0,
\]
then yields
\[
    \tilde R(t)\longrightarrow I_3.
\]
Therefore, $\tilde R=I_3$ is asymptotically stable, and its basin of attraction contains all initial conditions satisfying
\[
    \tr(\tilde R(0))
    \geq
    1+2\cos(\theta^\star).
\]

\subsection{Proof of Lemma~\ref{lemma:1b_2a}}
\label{proof: 1b 2a}
Define
\[
    \mathring b:=\frac{b}{|b|},
    \qquad
    \nu_a:=\frac{a_1\times a_2}{|a_1\times a_2|}.
\]
Since
\[
    S=bb^\top,
\]
one has
\[
    S^\dagger
    =
    \frac{\mathring b\mathring b^\top}{|b|^2},
    \qquad
    S^\dagger b
    =
    \frac{\mathring b}{|b|}.
\]
Moreover, since $a_1$ and $a_2$ are uniformly non-collinear,
\[
    \Lambda\Lambda^\dagger
    =
    \pi_{\nu_a}:=I_3-\nu_a\nu_a^\top.
\]
Using
\[
    (\Lambda^\top)^\dagger\Lambda^\top
    =
    \Lambda\Lambda^\dagger
\]
and $\tilde R=\hat R R^\top$, the innovation
\eqref{eq:innovation_brute} becomes
\begin{align}
    \Delta
    &=
    k[\mathring b]_\times
    \tilde R\pi_{R\nu_a}
    \big(\tilde R^\top\mathring b-\mathring b\big).
    \label{eq:innovation_1b2a}
\end{align}

Consider the Lyapunov function
\[
    V:=\tr(I_3-\tilde R).
\]
Using $[[x]_\times y]_\times=yx^\top-xy^\top$, one obtains
\begin{align}
    \dot V
    &=
    -\tr([\Delta]_\times\tilde R)\nonumber\\
    &=
    -k
    \big(\tilde R^\top\mathring b-\mathring b\big)^\top
    \pi_{R\nu_a}
    \big(\tilde R^{2\top}\mathring b-\mathring b\big).
    \label{eq:Vdot_1b2a_initial}
\end{align}

Define
\[
    E:=
    \mathring b\mathring b^\top
    -
    R\nu_a\nu_a^\top R^\top .
\]
Since $\mathring b$ and $R\nu_a$ are unit vectors,
\begin{align}
    \|E\|
    &=
    \left|
    \sin\!\left(
    \angle(\mathring b,R\nu_a)
    \right)
    \right|\nonumber\\
    &=
    \left|
    \sin\!\left(
    \angle(R^\top\mathring b,\nu_a)
    \right)
    \right|
    \leq\epsilon,
    \label{eq:E_bound_1b2a}
\end{align}
where the last inequality follows from
\eqref{condition:1b 2a epsilon}. Moreover,
\[
    \pi_{R\nu_a}
    =
    \pi_{\mathring b}+E.
\]
Hence,
\[
    \dot V=V_0+V_E,
\]
where
\begin{align}
    V_0
    &:=
    -k
    \big(\tilde R^\top\mathring b-\mathring b\big)^\top
    \pi_{\mathring b}
    \big(\tilde R^{2\top}\mathring b-\mathring b\big),
    \label{eq:V0_1b2a}\\
    V_E
    &:=
    -k
    \big(\tilde R^\top\mathring b-\mathring b\big)^\top
    E
    \big(\tilde R^{2\top}\mathring b-\mathring b\big).
    \label{eq:VE_1b2a}
\end{align}

These terms have exactly the same structure as those in the proof of
Lemma~\ref{lemma:2b 1a}, with $\tilde R^\top$ in place of $\tilde R$.
Let
\[
    \tilde R
    =
    I_3+\sin(\tilde\theta)[u]_\times
    +(1-\cos(\tilde\theta))[u]_\times^2,
    \quad
    \tilde\theta\in[0,\pi],\ u\in\mathbb S^2.
\]
Repeating the arguments of Appendix~\ref{proof:2b 1a} gives
\begin{align}
    V_0
    &\leq
    -k\cos(\tilde\theta)
    \left(
    1-\mathring b^\top
    \tilde R^2\mathring b
    \right),
    \label{eq:V0_bound_1b2a}\\
    |V_E|
    &\leq
    k\frac{\epsilon}{\cos(\tilde\theta/2)}
    \left(
    1-\mathring b^\top
    \tilde R^2\mathring b
    \right).
    \label{eq:VE_bound_1b2a}
\end{align}
Therefore,
\begin{align}
    \dot V
    \leq
    -k
    \left(
    1-\mathring b^\top\tilde R^2\mathring b
    \right)
    \left(
    \cos(\tilde\theta)
    -
    \frac{\epsilon}{\cos(\tilde\theta/2)}
    \right).
    \label{eq:Vdot_final_1b2a}
\end{align}

Consider an initial condition satisfying
\[
    \tr(\tilde R(0))
    \geq
    1+2\cos(\theta^\star),
\]
or equivalently
$\tilde\theta(0)\leq\theta^\star$.
For $\tilde\theta\in[0,\theta^\star]$,
\[
    \cos(\tilde\theta)
    -
    \frac{\epsilon}{\cos(\tilde\theta/2)}
    \geq
    \cos(\theta^\star)
    -
    \frac{\epsilon}{\cos(\theta^\star/2)}
    =:c^\star>0,
\]
where the strict positivity follows from
\eqref{condition:1b 2a epsilon}. Hence,
\begin{align}
    \dot V
    \leq
    -kc^\star
    \left(
    1-\mathring b^\top\tilde R^2\mathring b
    \right)
    \leq0.
    \label{eq:Vdot_negative_1b2a}
\end{align}

Since
\[
    V=2(1-\cos\tilde\theta),
\]
the set
\[
    \{\tilde R\in\SO(3):
    \tilde\theta\leq\theta^\star\}
\]
is forward invariant. Furthermore,
\[
    1-\mathring b^\top\tilde R^2\mathring b
    \in L_1.
\]
The boundedness and uniform continuity assumptions imply that this
quantity is uniformly continuous. Hence, by Barbalat's lemma,
\[
    1-\mathring b^\top
    \tilde R^2\mathring b
    \longrightarrow0.
\]
As in the proof of Lemma~\ref{lemma:2b 1a},
\[
    \left|
    \tilde R^\top\mathring b-\mathring b
    \right|^2
    =
    \frac{
    1-\mathring b^\top\tilde R^2\mathring b
    }{
    2\cos^2(\tilde\theta/2)
    },
\]
and therefore, since
$\tilde\theta\leq\theta^\star<\pi/2$,
\[
    \left|
    (\tilde R^\top-I_3)\mathring b
    \right|
    \longrightarrow0.
\]
It follows from \eqref{eq:innovation_1b2a} that
$\Delta\to0$, and consequently
$\dot{\tilde R}\to0$.

Finally, since $b$ is bounded, there exists $\bar b>0$ such that
$|b(t)|\leq\bar b$. Since
\[
    b=|b|\mathring b,
\]
one has
\[
    \mathring b\mathring b^\top
    =
    \frac{bb^\top}{|b|^2}
    \succeq
    \frac{1}{\bar b^2}bb^\top.
\]
Thus, persistence of excitation of $b$ implies persistence of
excitation of $\mathring b$. The standard persistence-of-excitation
argument of Proposition~4.6 in \cite{Trumpf2012}, together with
\[
    |(\tilde R^\top-I_3)\mathring b|\to0,
    \qquad
    \dot{\tilde R}\to0,
\]
then yields
\[
    \tilde R(t)\longrightarrow I_3.
\]
Therefore, $\tilde R=I_3$ is asymptotically stable, and its basin of
attraction contains all initial conditions satisfying
\[
    \tr(\tilde R(0))
    \geq
    1+2\cos(\theta^\star).
\]
\bibliographystyle{IEEEtran} 
\bibliography{references} 
\end{document}